\documentclass[runningheads]{llncs}
\usepackage{graphicx}
\usepackage[english]{babel}

\usepackage{amssymb,amsfonts,amsmath}
\usepackage{xcolor}
\usepackage{lineno}
\usepackage{fancybox}
\usepackage{times}

\usepackage{amsthm}
\usepackage{mathtools}
\DeclarePairedDelimiter\floor{\lfloor}{\rfloor}

\usepackage{microtype}

\usepackage{mathtools}
\usepackage{hyperref}
 \theoremstyle{plain}
 \newtheorem{observation}[theorem]{Observation}

\newcommand{\turnValue}[1]{\ensuremath{{D_\Omega}({#1})}}
\newcommand{\ang}[1]{\ensuremath{{\alpha}({#1})}}
\newcommand{\dir}[1]{\ensuremath{{d}({#1})}}
\newcommand{\dirl}[1]{\ensuremath{{d_\ell}({#1})}}
\newcommand{\dirr}[1]{\ensuremath{{d_r}({#1})}}

\def\obj{P}
\def\maxOmega{{\Omega^*}}

\title{Reconstructing a convex polygon from its $\omega$-cloud\thanks{
E. A. was partially supported by F.R.S.-FNRS, and by SNF  Early PostDoc Mobility project P2TIP2-168563.
P. B., J. C., and S. V. were partially supported by NSERC.}}

\author{
Elena Arseneva\inst{1} \and 
Prosenjit Bose\inst{2} \and
Jean-Lou De Carufel\inst{3} \and 
Sander Verdonschot\inst{2}
}
\authorrunning{Arseneva \and Bose \and De Carufel \and Verdonschot}

\institute{
St. Petersburg State University, 
Saint-Petersburg, Russia
\email{ea.arseneva@gmail.com}
\and
Carleton University, Ottawa, Canada\\
\email{jit@scs.carleton.ca}, \email{sander@cg.scs.carleton.ca}
 \and
University of Ottawa, Ottawa, Canada\\
\email{jdecaruf@uottawa.ca}
}

\begin{document}

\maketitle

\begin{abstract}
An $\omega$-wedge is the closed set of 
points contained between two rays that are emanating from a single point (the apex), and are separated by an angle $\omega < \pi$. Given a convex polygon $P$, we place the $\omega$-wedge such that $P$ is inside the wedge and both rays are tangent to $P$.
The set of apex positions of all such placements of the  $\omega$-wedge is called the \emph{$\omega$-cloud} of $P$. 



 We investigate reconstructing a polygon $P$ from its $\omega$-cloud.  
Previous work on reconstructing $P$ from probes with the $\omega$-wedge
required knowledge of the points of tangency between   
$P$ and the two rays of the $\omega$-wedge in addition to the location of the apex.
Here  we consider the setting where the \emph{maximal} $\omega$-cloud alone is given.  
We give two conditions under which  
it uniquely 
defines $P$: 
(i) when 
 $\omega < \pi$ is fixed/given, 
or (ii) 
when what is known is that $\omega < \pi/2$. 
We show that if neither of these two conditions hold, then $P$ may not be  unique. 
We show that, when the uniqueness conditions hold, the polygon $P$ can be reconstructed 
in $O(n)$ time with $O(1)$ working space in addition to the input, where $n$ is the number of arcs in the input $\omega$-cloud.  
 \end{abstract}

\section{Introduction}\label{sec:intro}

``Geometric probing considers problems of determining a geometric structure or
some aspect of that structure from the results of a mathematical or physical
measuring device, a probe.''~\cite[Page 1]{DBLP:journals/algorithmica/Skiena89}
Many probing tools have been studied in the literature such as finger probes~\cite{DBLP:journals/jal/ColeY87},
hyperplane (or line) probes~\cite{Dobkin:1986:PCP:12130.12174,DBLP:journals/ipl/Li88a},
diameter probes~\cite{DBLP:journals/ijrr/RaoG94},
x-ray probes~\cite{DBLP:journals/siamcomp/EdelsbrunnerS88,Gardner92},
histogram (or parallel x-ray) probes~\cite{MeijerSkiena96},
half-plane probes~\cite{DBLP:journals/jal/Skiena91} and
composite probes~\cite{DBLP:journals/amai/BrucksteinL91,DBLP:journals/ipl/Li88a} to name a few.
For example, diameter probes measure the width of the polygon at a certain direction; such measurements in all directions yield the diameter function of the polygon.  

A geometric probing problem can be considered as a {\em reconstruction} problem. Can one reconstruct an object given a set of probes?
For diameter probes this is not the case: there exists a class of (curved) \emph{orbiform}  shapes such as the celebrated \emph{Reuleaux triangle} that have the same diameter function as the circle.  
What is more surprising is that this is not the case even for polygonal shapes: there are uncountable families of polygons with the same diameter function, where the polygons  need not be regular, nor does their diameter function need to be constant~\cite{DBLP:journals/ijrr/RaoG94}.
In this paper we show, that an alternative probing device called an  \emph{$\omega$-wedge} yields a function that is free from these  drawbacks. 
 
The $\omega$-wedge, first studied by Bose et al.~\cite{DBLP:journals/comgeo/BoseCSS16},  is the closed set of 
points contained between two rays that are emanating from a single point, the \emph{apex} of the wedge. The angle $\omega$ formed by the two rays is such that $0 < \omega < \pi$.
A single probe of a convex $n$-gon $P$
is {\em valid} when 
$P$ is inside the wedge and both rays of the wedge are tangent to 
$P$, see Figure~\ref{fig:w-probe}a.\footnote{In~\cite{DBLP:journals/comgeo/BoseCSS16} probing with an $\omega$-wedge 
is defined for a wider class of  convex objects. Since here we focus on the objects being convex polygons, we restrict our definition accordingly.}
A valid probe returns the coordinates of the apex
and of the two points of contact between the wedge and the polygon. 
Using this tool,
a convex $n$-gon can be reconstructed using between $2n-3$ and $2n+5$ probes, depending on the value of $\omega$
and the number of 
\emph{narrow vertices} (vertices whose internal angle is at most $\omega$) in $P$~\cite{DBLP:journals/comgeo/BoseCSS16}. 
As an $\omega$-wedge rotates around $\obj$,
the locus of positions of the apex of the $\omega$-wedge
describes a curve called an \emph{$\omega$-cloud}, see Figure~\ref{fig:w-probe}c.
The $\omega$-cloud 
finds applications in diverse geometric algorithms~\cite{abellanas2011coverage,ALOUPIS2010115,moslehi2017separating}.

\begin{figure}[ht]
\begin{minipage}{0.32\linewidth}
\centering
\includegraphics{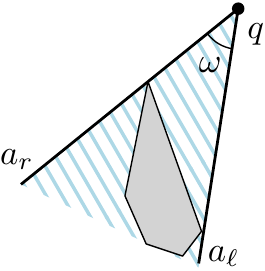} 
\\
(a)
\end{minipage}
\begin{minipage}{0.32\linewidth}
\centering
\includegraphics{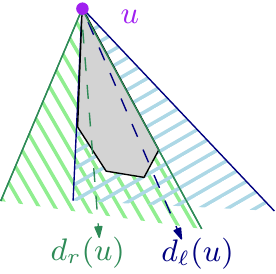}
\\
(b)
\end{minipage}
\begin{minipage}{0.32\linewidth}
\centering
\vspace{-1em}
\includegraphics[page=4,scale = 0.75]{narrow-pivot}
\\
(c)
\end{minipage}

\caption{A convex polygon $P$ (filled with gray color), and (a) A minimal $\omega$-wedge; (b) Narrow vertex $u$ of $P$, 
wedges $W_\ell(u)$ and $W_r(u)$ (shaded with  
rising  and 
falling tiling pattern, respectively) and their directions $\dirl{u}$ and $\dirr{r}$ (dashed lines); 
(c) The $\omega$-cloud $\Omega$ of $P$: the arcs (bold curved lines), pivots (filled disk marks), and all the supporting circles (light thin curved lines).}

\label{fig:w-probe}
\end{figure}

The $\omega$-cloud can be seen as a generalization of
the diameter function introduced by
Rao and Goldberg~\cite{DBLP:journals/ijrr/RaoG94}.
In their paper,
a diameter probe consists of
two parallel calipers turning around a convex object $\obj$ in the
plane. The function that returns the distance between these calipers as they
turn around $\obj$ is called a \emph{diameter function}.
Rao and Goldberg~\cite{DBLP:journals/ijrr/RaoG94}
show that two different
convex polygons can have the same diameter function, and the function need not be a constant. 
 This implies that
recovering the exact shape and orientation of a convex $n$-gon given only its
diameter function is not always possible and additional information is required.
An $\omega$-wedge can be seen as
two non-parallel calipers turning around a convex object $\obj$ in the plane.
As we prove in this paper, two different convex
 polygons cannot have the same (non-constant) $\omega$-cloud function, as opposed to the case of Rao and Goldberg's tool. 
This clearly shows the advantage of  the $\omega$-cloud against the latter tool.

Related to our probing method is also the method of Fleischer and Wang~\cite{DBLP:conf/isaac/FleischerW09}.  
In their method, a convex $n$-gon $\obj$ is placed inside a circle of
radius $1$. A camera that sees only the silhouette of $\obj$ can be placed
anywhere on the circle. The angle with which the camera sees $\obj$ together
with the position of the camera is a probe of $\obj$.  Let $\alpha$ be the
largest angle of $\obj$. They prove that if no two camera tangents overlap,
then $\left\lceil \frac{3\pi}{\pi-\alpha}\right\rceil$ probes are necessary and
sufficient. Otherwise, approximately $\left\lceil
\frac{4\pi}{\pi-\alpha}\right\rceil$ probes 
are sufficient. In our method, the apex of the $\omega$-wedge that turns
around $\obj$ can be seen as a camera. Thus our method is a variant of
theirs because instead of fixing the circle on which the camera can move, we fix
the angle from which $\obj$ can be seen by the camera.

In this paper, we prove that no two convex polygons have the same $\omega$-cloud (see Theorem~\ref{thm:uniqueness} in Section~\ref{sec:props}). 
We further consider a harder, but more natural problem of reconstructing a polygon from its \emph{maximal} $\omega$-cloud, which is a variant of an $\omega$-cloud where the consecutive arcs with the same supporting circle are merged in one arc (see Section~\ref{subsec:max-cloud}). 
We give two conditions under which the maximal $\omega$-cloud uniquely determines $P$: 
(i) if $\omega$ with $\omega < \pi$ is fixed/given, 
or (ii) without fixing $\omega$, but with a guarantee that $\omega < \pi/2$. 
We show that if neither of the conditions (i), (ii) hold, then $P$ may not be unique. 
Finally, we show that, when conditions  (i) or (ii) hold, the polygon $P$ can be reconstructed 
in $O(n)$ time and  $O(1)$ working space in addition to the input, 
where $n$ is the number of arcs in the input $\omega$-cloud (see Section~\ref{sec:reconstr}). 

The results in this paper differ from the ones in  Bose et al.~\cite{DBLP:journals/comgeo/BoseCSS16} 
in the following ways. 
First, we consider a different probing method, where   
only one piece of information is available for each probe instead of three, 
namely the coordinates of the apex of the $\omega$-wedge.
Second,  
we prove a fundamental property of the $\omega$-cloud, that 
no two convex polygons have the same 
$\omega$-cloud, and thus $\omega$-clouds yield a more advantageous probing method than diameter probes. 
Third, we derive a list of interesting properties of the $\omega$-cloud, which were not known before, 
and  that allow us to provide a simple algorithm to reconstruct 
the polygon from its maximal $\omega$-cloud. 
Note as well, that most of our results are more general than the ones from Bose et al.~\cite{DBLP:journals/comgeo/BoseCSS16}
since the results hold for any $0 < \omega < \pi$, instead of the more restricted range: $0 < \omega \leq \pi/2$.

\section{Properties of the $\omega$-cloud}\label{sec:props}

Let $P$ be an $n$-vertex convex polygon in $\mathbb{R}^2$.
For any vertex $v$ of $P$, let $\ang{v}$ be the internal angle of $P$ at $v$.
Let $\omega$ be a fixed angle with $0 < \omega < \pi$.
Consider an \emph{$\omega$-wedge} $W$; recall that it is the set of points contained between two rays $a_\ell$ and $a_r$ emanating from the same point $q$ (the \emph{apex} of $W$), 
such that the angle between the two rays is exactly $\omega$. See Figure~\ref{fig:w-probe}a.
We call the ray $a_\ell$ (resp., $a_r$) that bounds 
$W$ from the left (resp., right) as seen from $q$,  
 the \emph{left} (resp., \emph{right}) \emph{arm} of $W$. 
We say that an $\omega$-wedge $W$ is \emph{minimal} for $P$ if $P$ is contained in $W$ and the arms of $W$ are tangent to $P$.
The \emph{direction} of $W$ is given by the bisector ray of the two arms of $W$.
Note that for each direction, there is a unique minimal $\omega$-wedge.

\begin{definition}[$\omega$-cloud~\cite{bmss11}]
\label{def:w-cloud}
The \emph{$\omega$-cloud} of $P$ is the locus of the apices of all minimal $\omega$-wedges for $P$.
\end{definition}

Let $\Omega$ denote the $\omega$-cloud of $P$.
We define an \emph{arc} $\Gamma$ of $\Omega$ as a maximal connected portion of $\Omega$ such that for every point of $\Gamma$ the corresponding minimal $\omega$-wedge is combinatorially the same, i.e., its left and right arm touch the same pair of vertices of $P$.
Due to the inscribed angle theorem,  an arc of the $\omega$-cloud is a circular arc,
and $\Omega$ consists of a circular sequence of circular arcs, where each 
pair of consecutive arcs shares an endpoint, see also~Bose et al.~\cite{bmss11} for more details.
We refer to such a sequence as a circular arc sequence. 
The \emph{supporting circle} of a circular arc is the circle containing the arc. 
The points on $\Omega$ that are the intersection of two consecutive arcs of $\Omega$ 
are called \emph{pivots}.
Note that if $\omega \geq \pi/2$, two consecutive arcs of the 
$\omega$-cloud can lie on the same supporting circle. 
In this case we call the pivot separating them a \emph{hidden pivot}. For example, points $b, d, f$ in Figure~\ref{fig:j-lous}a are hidden pivots for the polygon $abcdef$, when $\omega = 5/6$. 
The total number of pivots (including hidden ones) is between $n$ and $2n$~\cite{bmss11}.

A vertex $v$ of $P$ is called \emph{narrow} if
the angle $\ang{v}$ 
is at most $\omega$. 
Note that a pivot of $\Omega$ coincides with a vertex of $P$ if and only if that vertex is narrow.
In this case, we also call such pivot \emph{narrow}.
If $\ang{v} < \omega$, we call $v$ (the vertex or the pivot) \emph{strictly narrow}.
The portion of $\Omega$ between 
two points $s,t \in \Omega$, unless explicitly stated otherwise, 
is the portion of $\Omega$
one encounters when traversing $\Omega$ from $s$ to $t$ clockwise, excluding $s$ and $t$.
We denote this portion of $\Omega$ by $\Omega_{st}$.
The \emph{angular measure} of an arc $\Gamma$ of $\Omega$ is the angle spanned by $\Gamma$, measured from the center of its supporting circle; 
the angular measure of a \emph{portion} of an arc is defined similarly.
For two points $s,t$ on $\Omega$, the \emph{total angular measure} of $\Omega$ from $s$ to $t$, denoted as 
$\turnValue{s,t}$, 
is the sum of the angular measures of all arcs in $\Omega_{st}$ (including the at most two non-complete arcs of $\Omega$).

Each point $x$ in the interior of an arc corresponds to a unique minimal $\omega$-wedge $W(x)$ 
with direction $\dir{x}$.
Let $u$ be a pivot of $\Omega$.
If $u$ is not strictly narrow, 
$u$ also corresponds to a unique minimal $\omega$-wedge $W(u)$ with  direction  $\dir{u}$. 
Otherwise, 
$u$ corresponds to a closed interval of directions $[\dirl{u}, \dirr{u}]$, 
where the angle between $\dirl{u}$ and $d_r(u)$ equals $\omega - \ang{u}$. See Figure~\ref{fig:w-probe}b. 
Let $W_\ell(u)$ and $W_r(u)$ denote the minimal $\omega$-wedges with apex at $u$ and directions respectively $d_\ell(u)$ and $d_r(u)$. 
For points $x$ on $\Omega$ that are not strictly narrow pivots, we define $\dirr{x}$ and $\dirl{x}$ both to be equal to $\dir{x}$, and  $W_\ell(x)$ and $W_r(x)$
equal to $W(x)$.

The now give a crucial property of the $\omega$-cloud that is the basis of the other properties.

\begin{lemma}
\label{lemma:ang-measure}
Let $s$ and $t$ be two points on $\Omega$ such that
$\Omega_{st}$ contains no narrow pivots. 
Then the angle between $\dirr{s}$ and $\dirl{t}$ is $\turnValue{s,t}/2$.
\end{lemma}
\begin{proof}
Suppose first that $\Omega_{st}$ is a single arc, see Figure~\ref{fig:left-right-points}a.
The angle $\alpha$ between $\dirr{s}$ and $\dirl{t}$ equals the angle $\beta$ between the left arms of the two minimal $\omega$-wedges corresponding to these directions, which in turn equals angle $\beta'$. 
By the inscribed angle theorem, $\beta'$ is half the central angle spanning the arc $\Omega_{st}$, which, by definition of the angular measure, equals $\turnValue{s,t}/2$.

Now suppose that $\Omega_{st}$ consists of several arcs, which are either arcs of $\Omega$ or (at most two) connected portions of those arcs. 
By assumption, none of the pivots separating these arcs are narrow (although $s$ and $t$ may be, see the definition of $\Omega_{st}$).
Consider the change of direction 
of the minimal $\omega$-wedge as its apex moves from $s$ to $t$. After traversing an arc with endpoints $u$ and $v$, 
by the above observation, the direction changes by exactly $\turnValue{u,v}/2$.
Since none of the 
pivots on $\Omega_{st}$ are narrow, there is no change in direction as the $\omega$-wedge passes through each pivot.
Therefore the total change in direction for the $\omega$-wedge is the sum of the changes induced by the traversed arcs, that is, $\turnValue{s,t}/2$.
\end{proof}


\begin{corollary}
  \label{cor:total-measure}
  The sum of the angular measures of all arcs of $\Omega$ 
  is $2(2\pi-\sum_{v \in S}(\omega-\ang{v}))$, where $S$ is the set of all
  narrow vertices of $P$.  
  In particular, 
 $P$ has no strictly narrow vertices if and only if
  the sum of angular measures of the arcs of $\Omega$ is $4\pi$.   
\end{corollary}

\begin{proof}
Choose a point $x$ in the interior of an arc of  $\Omega$ and consider the change in the direction of the minimal $\omega$-wedge as its apex traverses the whole $\Omega$ and returns back to $x$. Since the minimal $\omega$-wedge at $x$ is unique,  the total angle the direction has turned is $2\pi$. 
The narrow pivots break $\Omega$ in a number of connected components, for which Lemma~\ref{lemma:ang-measure} applies. 
The total angular measure of the arcs of $\Omega$ is the sum of angular measures of these components. 
In a narrow pivot $u$,  
the minimal $\omega$-wedge turns from the position when its left arm coincides with the polygon edge incident to $u$ and following $u$ in the clockwise direction, to the position when its right arm coincides with the polygon edge 
preceding $u$. See Figure~\ref{fig:w-probe}b. 
Thus the angle the direction of the minimum $\omega$-wedge turned, while its apex has stayed in $u$, equals $\omega-\ang{u}$. 
Summing up such  deficit for all narrow pivots, we obtain that the total turn of the $\omega$-wedge corresponding to the arcs of $\Omega$ is 
$2\pi - \sum_{v \in S}(\omega-\ang{v})$. Lemma~\ref{lemma:ang-measure} as applied to the components of $\Omega$ as separated by the narrow pivots,  implies that 
the total angular measure of all arcs of $\Omega$ equals $2(2\pi - \sum_{v \in S}(\omega-\ang{v}))$.
\end{proof}

Since each arm of any minimal $\omega$-wedge touches $P$, and since $P$ is convex,  Lemma~\ref{lemma:ang-measure} implies the following. 

\begin{corollary}
\label{cor:max-measure}
For any arc $\Gamma$ of $\Omega$, 
 the angular measure of $\Gamma$ is at most $2(\pi-\omega)$.  
\end{corollary}

The following is a simple fact directly implied by the definition of the $\omega$-cloud. 

\begin{lemma}
\label{lemma:normal-pivot}
Let $u$ be a non-narrow pivot of $\Omega$.  Let $x$ be the second point of intersection between the supporting circles of the two arcs of $\Omega$ adjacent to $u$ (the first point of intersection is $u$ itself).  The minimal $\omega$-wedge with the apex at $u$ touches $x$ with one of its arms. \end{lemma}

\begin{figure}[ht]
\begin{minipage}{0.32\linewidth}
\centering
\includegraphics{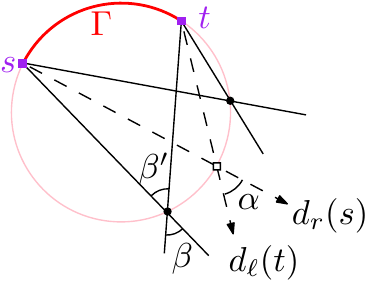}
\vspace{0.9em}
\\
(a)
\end{minipage}
\begin{minipage}{0.32\linewidth}
\centering
\includegraphics[page=2,scale = 0.83]{narrow-pivot}
\\
(b)
\end{minipage}
\begin{minipage}{0.32\linewidth}
\centering
\includegraphics[page=1, scale = 0.83]{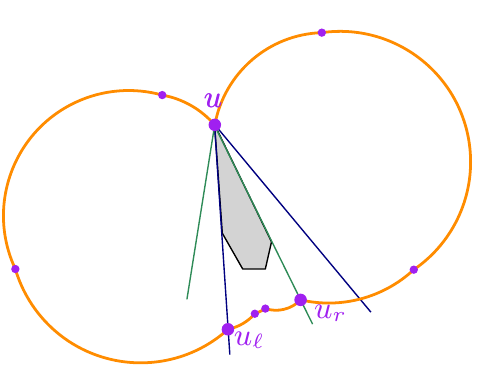}
\\
(c)
\end{minipage}
\caption{
(a) Illustration for the proof of Lemma~\ref{lemma:ang-measure}, the case when $\Omega_{st}$ is a single arc. Vertices of $P$ are marked as filled disks.
(b) Point $x$ in the interior of an arc of $\Omega$, 
wedge $W(x)$ with its direction $\dir{x}$, 
points $x_\ell$ and $x_r$. (c) Narrow pivot $u$, 
wedges $W_\ell(u)$ and $W_r(u)$,  points $u_\ell$ and~$u_r$.}
\label{fig:left-right-points}
\end{figure}


\textbf{Characterization of narrow pivots.}
Let $x$ be a point on $\Omega$. It corresponds to two  minimal $\omega$-wedges $W_\ell(x)$, $W_r(x)$ (which coincide if $x$ is not a narrow pivot). 
Consider the open ray of the  right arm of $W_\ell(x)$. 
A small neighborhood of its origin is in the interior of the region bounded by $\Omega$,
 thus the ray intersects $\Omega$ at least once. 
Among the points of this intersection, let $x_\ell$ be the one closest to $x$. Define the point $x_r$ analogously for the left arm of $W_r(x)$. 
See Figure~\ref{fig:left-right-points}b,c. 

\begin{lemma}
\label{lemma:turns}
(a) Neither $\Omega_{x_{\ell}x}$ nor $\Omega_{xx_r}$ contains a narrow pivot. 
(b) If $x$ is a narrow pivot, then $\turnValue{x_\ell,x} = \turnValue{x,x_r} = 2(\pi - \omega)$. 
(c) If $x$ is not narrow, 
then either $\turnValue{x,x_r} = 2(\pi - \omega)$, or $x_r$ is the first narrow pivot following $x$ in the clockwise direction. A symmetric statement holds for $x_\ell$.
\end{lemma}

\begin{proof}
We prove 
the statement for $x$ and $x_r$; 
 symmetric arguments apply for $x_\ell$ and $x$. 
 
\textit{(a)} If $\Omega_{xx_r}$ contained a narrow pivot $u$, then $u$ had to lie inside the wedge  
$W(x)$ (because it would be a vertex of $P$). However, in that case 
the left arm of $W(x)$ should have intersected $\Omega$ before $x_r$, contradicting to the definition of $x_r$. 

\textit{(b)} By Lemma~\ref{lemma:ang-measure}, $\turnValue{x,x_r}$ is twice the angle between $\dirr{x}$ and $\dirl{x_r}$.
Thus we need to show that the latter is $\pi - \omega$. 
By definition, $x_r$ lies on the line $\ell$ through the left arm
of the wedge $W_r(x)$. Since $x$ is narrow, $\ell$ passes through an edge of $P$. 
Now consider $W_\ell(x_r)$, the leftmost minimal $\omega$-wedge corresponding to $x_r$. 
Its right arm is aligned with $\ell$, as this is the leftmost possible line through $x_r$ 
with $P$ entirely to its left. See Figure~\ref{fig:left-right-points}c. 
Therefore the angle between the bisectors of $W_r(x)$ and $W_\ell(x_r)$ is $\pi-\omega$, and $\turnValue{x,x_r}$ is $2(\pi-\omega)$. 

\textit{(c)} Let $x$ be not a narrow pivot. We suppose $\turnValue{x,x_r} \neq 2(\pi - \omega)$, and will show that $x_r$ is a narrow pivot of $\Omega$. By the above observation that there can be no narrow pivots between $x$ and $x_r$, this will imply the claim.  
Suppose $x_r$ is not a narrow pivot. 
Then there is a vertex $v \neq x_r$ of $P$, 
 where the left arm of $W(x)$ touches $P$. 
Then the right arm of $W_\ell(x_r)$ must be aligned with $\ell$, which would mean $\turnValue{x,x_r} = 2(\pi - \omega)$. 
A contradiction. 
\end{proof}

The above lemma characterizes narrow pivots in terms of the position of the corresponding minimal $\omega$-wedges. 
Now we characterize them in terms of the adjacent arcs and their measures, the information that can be used by 
the reconstruction algorithm. 

\begin{lemma} 
\label{lemma:narrow-char}
Let $u$ be a pivot of $\Omega$, and let $v$ and $w$ be the points on $\Omega$ such that 
$\turnValue{v,u} = \turnValue{u, w} = 2(\pi - \omega)$. \\ 
 (a) If pivot $u$ is narrow, then 
the supporting circles of all the arcs of $\Omega_{vw}$ pass through $u$.\\
(b) Pivot $u$ is narrow, if at least one of the following conditions is satisfied: 
(i) $\Omega_{vu}$  consists of a single arc;  (ii) there is an  arc $\Gamma$ of 
$\Omega_{vu}$ that is not incident to $u$, such that the supporting circle of $\Gamma$ contains $u$. 
\\
 A symmetric statement holds for $\Omega_{uw}$.
\end{lemma}

\begin{proof}
(a) Suppose that $u$ is a narrow pivot. 
Recall the definition of $u_\ell, u_r$ from above, see also Figure~\ref{fig:left-right-points}c.
 By 
Lemma~\ref{lemma:turns}b the points $v$ and $w$ coincide with 
$u_\ell$ and $u_r$, respectively. See Figure~\ref{fig:circles}. The minimal $\omega$-wedge, as its apex traverses $\Omega_{u_{\ell}u}$, is always touching $u$ with its left arm.
Indeed, when its apex is at $u_\ell$, its left arm is aligned with the line through $u$ and $u_\ell$, and thus with the edge $e_\ell$ of $P$ incident to $u$ and preceding it in clockwise order.
When the apex 
reaches $u$, its direction is $\dirl{u}$ and its 
right arm is touching $e_\ell$. Thus every placement of the minimal $\omega$-wedge in the considered interval was touching $u$ with its left arm. Symmetrically, every minimal $\omega$-wedge
 with the apex at any point between $u$ and $u_r$ is touching $u$ with its right arm. 
The claim is implied.

(b) Suppose first that  $\Omega_{vu}$ consists of a single arc $\Gamma$. 
We show that any minimal $\omega$-wedge corresponding to this arc must pass through both $u$ and $v$, which means that they both are vertices of $P$, and thus narrow pivots.  To see this, pick any point $y$ on $\Gamma$, and any point
$z$ on the supporting circle of $\Gamma$ outside $\Gamma$. See Figure~\ref{fig:b}. The quadrilateral $zvyu$ is inscribed in the circle, thus the sum of the angle at $z$ and the one at $y$ is $\pi$ (the angles $\beta$ and $\gamma$ in  Figure~\ref{fig:b}). The former angle is half the angular measure of $\Gamma$ (the angle $\beta$ in Figure~\ref{fig:b} is half the angle $\alpha$), i.e., it is $\pi-\omega$. Therefore the angle at $y$ is $\omega$, i.e., the minimal $\omega$-wedge 
with the apex $y$ is passing through $u$ and $v$.

\begin{figure}
\begin{minipage}{0.49\linewidth}
\centering
\includegraphics[page=1]{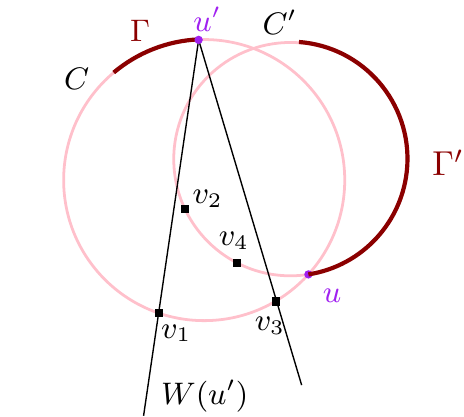}
\\ (a)
\end{minipage}
\begin{minipage}{0.49\linewidth}
\centering
\includegraphics[page=2]{narrow-char-proof-1}
\\ (b)
\end{minipage}

\caption{Illustration for the proof of Lemma~\ref{lemma:narrow-char}b,  the case when the portion of $\Omega$ between $v$ and $u$ consists of several  arc.}
\label{fig:narrow-char}
\end{figure}

Suppose  now that $\Omega_{vu}$ contains at least two arcs, and suppose that there exists arc  $\Gamma \in \Omega_{vu}$ that is not incident to $u$, but whose supporting circle $C$ passes through $u$, see Figure~\ref{fig:narrow-char}. 
Let $\Gamma'$ be the arc of $\Omega_{vu}$ incident to $u$, and let $C'$ be the supporting circle of $\Gamma'$.  
Let $u'$ be the most clockwise endpoint of $\Gamma$. 
It is enough to show that pivot $u'$ corresponds to the turn of the $\omega$-wedge around $u$: indeed, this would imply 
that $u$ coincides with a vertex of $P$, and thus it is a narrow pivot. 

Assume for the sake of contradiction that the above does not hold.
Suppose first that the pivot $u'$ is not narrow. 
Observe that the 
wedge $W(u')$ cannot intersect arc $\Gamma$ with its arms. 
Thus the arms of $W(u')$ touch circle $C$ in the apex $u'$ and two points $v_1, v_3$, that are the two vertices of $P$ that correspond to the arc $\Gamma$. See Figure~\ref{fig:narrow-char}a. 
The two vertices $v_2,v_4$ corresponding to the arc $\Gamma'$ 
 must lie inside $W(u')$.  
By construction, $v_1,v_2,v_3,v_4$ must be a  subsequence of the sequence of vertices of the polygon $P$ 
in clockwise order. 
However, by the assumption of the lemma, both $C$ and $C'$ pass through $u$ and $u'$, 
and thus points  $v_1,v_2,v_3,v_4$ (in this order)  are not in convex position, contradicting to $P$ being convex. 

Suppose now that the pivot $u'$ is narrow. The minimal $\omega$-wedge as its apex traverses arc $\Gamma$ 
 turns around $u'$ and some vertex $v_1$ of $P$. See Figure~\ref{fig:narrow-char}b.
The minimal $\omega$-wedge $W(u)$ of $u$ must contain $u'$ since it is a vertex of $P$.  
For example, in Figure~\ref{fig:narrow-char}b the right arm of this wedge passes through $u'$.
However, in this case one of the vertices of $P$ that correspond to the arc 
$\Gamma'$ is outside $W(u')$, see the unfilled circle mark in Figure~\ref{fig:narrow-char}b. A contradiction. 

The proof of item (b) for $\Omega_{uw}$ is analogous. 
\end{proof}

\begin{figure}
\centering
\includegraphics{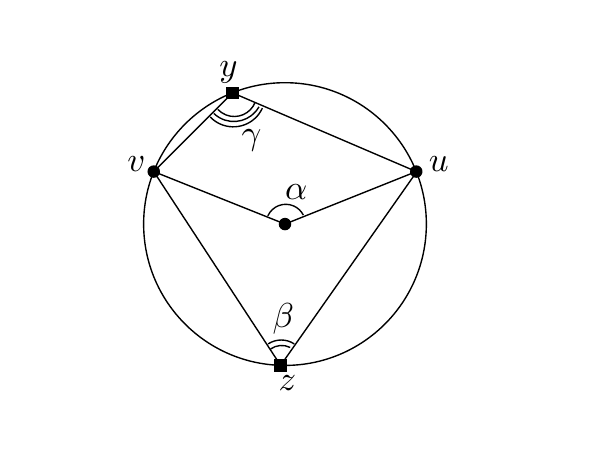}
\caption{ Proof of Lemma~\ref{lemma:narrow-char}b, the case when the portion of $\Omega$ between $v$ and $u$ is a single arc.}
\label{fig:b}
\end{figure}

\begin{observation}
\label{obs:narrow}
If $u$ is a narrow pivot followed by a non-narrow pivot, then the condition (ii) of item (b) of Lemma~\ref{lemma:narrow-char} for $\Omega_{uw}$ is satisfied, 
and 
the supporting circle of the arc following $w$ does not pass through $u$. 
\end{observation}

\begin{proof}
Let $e = uu'$ be the edge of $P$ adjacent to $u$ and following it in the clockwise direction.  
Consider the  motion of the minimal $\omega$-wedge as its apex goes from $u$ to $w$. In the beginning, it rotates around $u$ and $u'$. 
The first pivot $w'$ the apex reaches corresponds to the change of one of those points.
If $u$ is changed, then the right arm of the wedge must be collinear with $e$. In this case $w' = u' = w$, and this is a narrow pivot, which contradicts the assumption of this observation. 
 If $u'$ is changed, then  the condition (ii) of item (b) of Lemma~\ref{lemma:narrow-char} 
is satisfied for $\Omega_{uw}$  and the supporting circle of the arc following pivot $w'$ is passing through $u$ and some other vertex of $P$. 

Let $\Gamma'$ be the arc following $w$. Since $\turnValue{u,w} = 2(\pi-\omega)$, the minimal $\omega$-cloud with the apex at $w$
has its right arm collinear with $e$.  If the supporting circle of $\Gamma'$ was passing through $u$, then the minimal $\omega$-wedge $W(a)$
 with the apex at any point $a$ in the interior of  $\Gamma'$ would touch $u$  with its right arm. But then the point $u'$ is outside 
 $W(a)$, yielding a contradiction.  
\end{proof}

\textbf{Characterization of hidden pivots.}
Recall that a pivot is called hidden if its incident arcs have the same supporting circle.

\begin{lemma}
\label{lemma:hidden-pivots}
Let $u$ be a hidden pivot of $\Omega$, let $\Gamma_\ell$ and $\Gamma_r$ be the two arcs of $\Omega$ incident to $u$, and let $v$ and $w$ be the other endpoints of $\Gamma_\ell$ and $\Gamma_r$, respectively.
Then $v$, $u$, and $w$ are all narrow and each of the arcs $\Gamma_\ell$, $\Gamma_r$ has angular measure $2(\pi-\omega)$.
\end{lemma}

\begin{proof}
Since $u$ is a hidden pivot,  $\Gamma_\ell$ and $\Gamma_r$ are supported by the same circle $C$. Consider the minimal $\omega$-wedge as its apex traverses $\Gamma_\ell$. Its arms are touching two vertices of $P$, both lying on $C$. 
Since $u$ is a hidden pivot,  the vertex touched by the left arm of the wedge is $u$ (otherwise, there would be no possibility to switch from $\Gamma_\ell$ to $\Gamma_r$ at point $u$). 
Let $a$ be the vertex of $P$ touched by the right arm of the wedge. When the apex of the wedge reaches $u$, the wedge becomes $W_\ell(u)$, and its right arm is passing through $a$ and $u$. If polygon $P$ had a vertex between $a$ and $u$, that vertex must lie outside $W_\ell$, which is impossible. Thus $a$ is actually $v$. Therefore, $\turnValue{v,u} = \turnValue{a,u}$, which by  
Lemma~\ref{lemma:turns}b equals $2(\pi-\omega)$. A symmetric argument for $u$, $v$, and $\Gamma_r$ completes the proof.
\end{proof}

We proceed with the main result of this section.
\begin{theorem}
\label{thm:uniqueness}
For a circular arc sequence  $\Omega$, there is at most 
 one convex polygon $P$ and one angle $\omega$ with $0 < \omega< \pi$ such that  $\Omega$ is the $\omega$-cloud of $P$. 
\end{theorem}

\begin{proof}

We prove that the angle $\omega$ is unique. Suppose for the sake of contradiction that it is not the case, i.e., there are two distinct angles $\omega, \omega'$, for which 
$\Omega$ is the $\omega$-cloud and the $\omega'$-cloud of some polygons.   Assume without loss of generality that $\omega < \omega'$. 

Suppose first that the total angular measure of $\Omega$ is  $4\pi$. By Corollary~\ref{cor:total-measure}, it does not have any narrow pivots. 
Fix any pivot $u$ of $\Omega$. Lemma~\ref{lemma:turns}c  uniquely determines the placement of the minimal $\omega$- and $\omega'$-wedge with the apex at point $u$: 
 The two arms of the $\omega$-wedge must pass through the two points on $\Omega$ at distance $2(\pi-\omega)$ from $u$; analogously for the $\omega'$. By this construction, the  $\omega'$-wedge must enclose the  $\omega$-wedge. 
But both wedges must touch the intersection point of the two supporting circles of the arcs incident to $u$ due to Lemma~\ref{lemma:normal-pivot}.
This is possible only if the left or the right arms of the two wedges coincide. But then by Lemma~\ref{lemma:turns}c, as applied to the $\omega$-wedge, the point on $\Omega$ at the distance $2(\pi-\omega)$ is a narrow pivot of $\Omega$. This contradicts our assumption. 

Suppose now that the total angular measure of $\Omega$ is less than $4\pi$. By Corollary~\ref{cor:total-measure}, there are narrow pivots in $\Omega$. Moreover, if a pivot is narrow for $\omega$, it must be narrow for $\omega'$.
If all the pivots in $\Omega$ are narrow for $\omega$, they are also all narrow for $\omega'$, moreover, they are exactly the vertices of $P$, but then the same polygon $P$ has the $\omega$-cloud and $\omega'$-cloud coinciding with $\Omega$, which contradicts the definition of the $\omega$-cloud. Therefore, not all the pivots of $\Omega$ are narrow for $\omega$. 

 Let $u$  be a narrow pivot for $\omega$ followed by a non-narrow pivot for $\omega$; it exists because of our above observation that not all pivots of $\Omega$ are narrow for $\omega$. Let $u, w, u', w'$ be the points of $\Omega$ with $\turnValue{u,v} = \turnValue{v,w} = 2(\pi-\omega)$ and $\turnValue{u',v} = \turnValue{v,w'} = 
2(\pi-\omega')$.  
By Lemma~\ref{lemma:narrow-char}, the supporting circles of  all the arcs
in $\Omega_{uw}$ pass through $u$. Since $\omega < \omega'$, point $w'$ lies in the interior of $\Omega_{uw}$. Thus the arc $\Gamma$ that follows $w'$ must have the supporting circle passing through $u$.  
But by Observation~\ref{obs:narrow}, the supporting circle of $\Gamma$ does not pass through $u$. 
This is a contradiction.

We showed that the angle $\omega$ is unique. We now show that the polygon $P$, for which $\Omega$ is the $\omega$-cloud, is unique as well. 
Observe that $\Omega$ must have at least two arcs.
 Lemmas~\ref{lemma:narrow-char} and~\ref{lemma:hidden-pivots} uniquely identify the narrow pivots of $\Omega$, which are the narrow vertices of $P$ (including hidden narrow pivots).
 By Lemma~\ref{lemma:turns}a, the portion of $\Omega$ between any two narrow pivots has total angular measure at least 
$2(\pi-\omega)$, and thus the components of $P$ as defined by excluding all narrow vertices  are uniquely determined by Lemma~\ref{lemma:turns}b. In particular, for each such component  Lemma~\ref{lemma:turns}b gives the minimal $\omega$-wedge $W$ with the apex at some point $x$ in that component. Let $\Gamma$ be an arc that contains or is incident to $x$, and $C$ be the supporting circle of $\Gamma$. Intersecting $W$ with $C$ gives the two vertices of $P$, that are tangent to the arms of the $\omega$-wedge as its apex traverses $\Gamma$. 
\end{proof}

\subsection{Maximal $\omega$-cloud}
\label{subsec:max-cloud}
Given the $\omega$-cloud $\Omega$ of $P$, let the \emph{ maximal $\omega$-cloud} of $P$, denoted $\maxOmega$, be the result of merging all the pairs of consecutive arcs in $\Omega$ that have same supporting circle; equivalently $\maxOmega$ can be seen as the result of removing all the hidden pivots from $\Omega$. 

The maximal $\omega$-cloud is a natural modification of the $\omega$-cloud for the reconstruction task, see Section~\ref{sec:reconstr}; it reflects the situation when as an input we are given a locus of the apices of all the minimal $\omega$-wedges without any additional information, 
i.e., without the coordinates of the hidden pivots. 
The following  two statements are corollaries of Lemma~\ref{lemma:hidden-pivots}.

\begin{lemma}
\label{cor:circular-cloud}
The maximal $\omega$-cloud of $P$ is a circle $C$ if and only if
$k = \pi/(\pi-\omega)$ is an integer and $P$ is a regular $k$-gon inscribed in $C$.  
\end{lemma}
\begin{proof}
Suppose the maximal $\omega$-cloud $\maxOmega$ of $P$ is a circle $C$. Since all arcs of 
$\maxOmega$ of $P$ are supported by the circle $C$, all the pivots of $\Omega$ are hidden. Applying Lemma~\ref{lemma:hidden-pivots} to each pivot, we obtain that the pivots of $\Omega$ are exactly the vertices of $P$, and each arc has measure $2(\pi-\omega)$. Therefore the number of vertices of $P$ is $2\pi/(2(\pi-\omega))$, and the claim follows.
The other direction of the claim directly follows from 
the definition of an  $\omega$-cloud. 
\end{proof}

\begin{lemma} 
\label{lemma:concat}
 An arc of $\maxOmega$ is greater than $2(\pi-\omega)$ if and only if this arc is a concatenation of
 at least two  co-circular arcs of $\Omega$ separated by hidden pivots, each of measure $2(\pi-\omega)$.
\end{lemma}

Lemmas~\ref{cor:circular-cloud} and~\ref{lemma:concat} 
together with Theorem~\ref{thm:uniqueness} imply a similar uniqueness result for  $\maxOmega$: 
 
\begin{theorem}
\label{thm:uniqueness-max}
 Given an angle $\omega$, and a circular arc sequence  $\maxOmega$ of at least two arcs, there is at most 
 one convex polygon $P$ such that  $\maxOmega$ is the maximal $\omega$-cloud of $P$. 
\end{theorem}

Another uniqueness result, which is useful
for our reconstruction algorithm, follows from   Theorem~\ref{thm:uniqueness} and the
fact that for $\omega < \pi/2$ the $\omega$-cloud and the maximal $\omega$-cloud coincide. 
  
\begin{corollary} [of Theorem~\ref{thm:uniqueness}]
For a circular arc sequence  $\Omega$, there is at most 
 one convex polygon $P$ and one angle 
 $\omega< \pi/2$ such that  $\Omega$ is the maximal $\omega$-cloud of $P$.
\end{corollary}

We now show that the condition that $\omega< \pi/2$ is necessary for the above statement. 

\begin{proposition}
\label{prop:non-unique}
There is a circular arc sequence such that it is the maximal $\omega$-cloud of $P$ and the maximal $\omega'$-cloud of $P'$ for distinct  angles $\omega, \omega'$ and polygons $P,P'$. 
\end{proposition}
\begin{proof}
Let $P$ be a triangle with vertices $a,c,e$,  such that $\angle ace = \pi/3$,  $\angle cea = 16/45\pi$, 
and $\angle cae = 14/45\pi$. See Figure~\ref{fig:j-lous}a.  Let $\omega = 2/3\pi$, and let $\Omega$ be 
the $\omega$-cloud of $P$; it consists of three circular arcs $\Gamma_0, \Gamma_1,$ and $\Gamma_2$; note that the supporting circles of these arcs are pairwise distinct. Since $\Omega$ has no hidden pivots, $\Omega$ is also the maximal $\omega$-cloud of $P$. Let $b$ (resp.,
$d$ and $f$) be the midpoint of $\Gamma_0$ (resp.,  of $\Gamma_1$ and of $\Gamma_2$), and let $P'$ be the hexagon $abcdef$. 
Then $\angle abc = \angle cde =  \angle efa =  \angle bcd = 2/3 \pi$, $ \angle fab = 29/45 \pi$, and  $ \angle fab = 29/45 \pi$.
Let $\omega' = 5/6\pi$, and note that all the angles at the vertices of $P'$ are smaller than $\omega'$. 
It is easy to verify that $\Omega$ is the maximal $\omega'$-cloud of $P'$. 
\end{proof}

\begin{figure}
\begin{minipage}{0.49\linewidth}
\centering 
\includegraphics{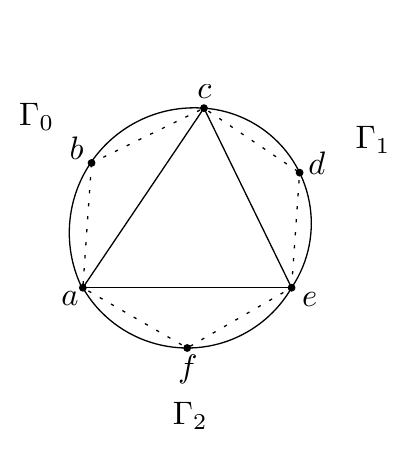}
\\
(a)
\end{minipage}
\begin{minipage}{0.49\linewidth}
\centering 
\includegraphics[page=5]{narrow-pivot}
\\
(b)
\end{minipage}
\caption{(a) A circular arc sequence of $3$ arcs, that is the maximal $2/3\pi-$cloud for triangle $ace$ and the maximal $5/6-$cloud for hexagon $abcdef$. (b) Illustration for the proof of Lemma~\ref{lemma:reconstr-no-narrow}.}
\label{fig:j-lous}

\end{figure}

\begin{remark}
The example in the proof of Proposition~\ref{prop:non-unique} can be continued by subdividing each arc in half and considering $11/12\pi$-cloud for the obtained $12$-gon, and so on.
\end{remark}

\section{Reconstructing $P$ from its maximal \texorpdfstring{$\omega$}{omega}-cloud}\label{sec:reconstr}
In this section we let 
$\maxOmega$ be a circular arc 
 sequence.  
Our goal is to reconstruct the convex polygon $P$ for which $\maxOmega$ is the maximal $\omega$-cloud for some angle $\omega$, or determine that no such polygon exists.

Note that if $\maxOmega$ is a single arc, i.e., it is a circle $C$, then $P$ is not unique. By Lemma~\ref{cor:circular-cloud}, it is a regular $\pi/(\pi-\omega)$-gon inscribed in $C$. The position of its vertices on $C$ is impossible to identify given only $\maxOmega$ and $\omega$. 
Therefore we assume that $\maxOmega$ has at least two arcs.

First, in Section~\ref{sec:reconstr-given-w}, we consider $\omega$ to be given, and such that $0 < \omega < \pi$; in this case the maximal $\omega$-cloud of $P$ may differ from its $\omega$-cloud.  
Afterwards, in Section~\ref{sec:reconstr-without-w}, we consider the setting where $\omega$ is not known, but it is known that  $0 < \omega < \pi/2$; in this case the two variations of the $\omega$-cloud coincide, and we use $\Omega$ instead of $\Omega^*$ to denote the input.

\subsection{ An  \texorpdfstring{$\omega$}{omega}-aware reconstruction algorithm}
\label{sec:reconstr-given-w}

Here we assume that $\omega$ is given, and that $0 < \omega < \pi$. 
The main difficulty in the reconstruction task is caused by  strictly narrow vertices, as the turn of the minimal $\omega$-wedge at those vertices is not reflected in the $\omega$-cloud, see Corollary~\ref{cor:total-measure}.
In our reconstruction algorithm, we first find all the non-hidden
strictly narrow pivots of $\maxOmega$, and then  
treat the connected portions of $\maxOmega$ between those pivots separately. 
Lemma~\ref{lemma:reconstr-no-narrow} gives a procedure to process such a portion. 

For two points $u$ and $v$ on $\maxOmega$, let $P_{uv}$ be the union of the edges and vertices of $P$ touched by the arms of the minimal $\omega$-wedge as its apex traverses $\Omega^*_{uv}$. Since $\Omega^*_{uv}$ is connected, when the apex of the wedge traverses $\Omega^*_{uv}$, for each of its two arms, the point of contact with $P$ traverses a connected portion of $P$; thus  $P_{uv}$ consists of at most two connected portions of $P$. Note that it is possible that one of the portions is a single vertex.

\begin{lemma}
\label{lemma:reconstr-no-narrow}
Given a portion  $\Omega^*_{uv}$ of $\maxOmega$ that 
does not contain any strictly narrow pivots,  
and the direction $\dirr{u}$ of the rightmost minimal $\omega$-wedge $W_r(u)$ of $u$,  
the portion $P_{uv}$ of $P$ that corresponds to $\Omega^*_{uv}$  
can be reconstructed 
in time linear in the number 
of arcs in $\Omega^*_{uv}$. The procedure requires $O(1)$ working space in addition to the input.  
\end{lemma}
\begin{proof} 
 Let  $\Gamma$ be the arc of $\Omega^*_{uv}$ incident to $u$, let $u'$ be the other endpoint of $\Gamma$,  and let $C$ be the supporting circle of $\Gamma$.
 See Figure~\ref{fig:j-lous}b.
 By knowing the value of $\omega$ and the direction 
$\dirr{u}$, 
we 
determine the wedge $W_r(u)$. 
The intersection between the wedge $W_r(u)$ and the circle  $C$ 
determines the 
two vertices of $P$ touched by the minimal $\omega$-wedge as its apex traverses arc $\Gamma$.
In Figure~\ref{fig:j-lous}b, these two points are $u$ and $p$.
The direction of the leftmost minimal $\omega$-wedge at $u'$, $W_\ell(u')$, is   $\dirr{u} + \turnValue{u, u'}/2$ due to Lemma~\ref{lemma:ang-measure}.
If $u'$ is inside $\Omega^*_{uv}$, then $u'$ is not a strictly narrow vertex, and thus
there is a unique  minimal $\omega$-wedge $W(u')$ 
at $u'$, $W(u') = W_\ell(u')$. 
Therefore, for each arc of $\Omega^*_{uv}$ we find the pair of vertices of $P$ that induces that arc. 
Moreover,  by visiting  the pivots of $\Omega^*_{uv}$ one by one, we find the vertices of each of the two chains of $P_{uv}$ ordered clockwise. 
To avoid double-reporting vertices of $P$, we keep the startpoints of the two chains, and whenever one chain reaches the startpoint of the other one, we stop reporting the points of the former chain. 
  
This procedure visits the pivots of $\Omega^*_{uv}$ one by one, and performs $O(1)$ operations at each pivot, namely, finding the intersection between a given wedge and a given circle. No additional information needs to be stored.
\end{proof}

\subsubsection{Reconstruction algorithm.} As an input, we are given an angle  $\omega$, $0 < \omega < \pi$, 
and a circular arc sequence $\maxOmega$ which is not a single circle.   
We now describe an algorithm to check if $\maxOmega$ is the $\omega$-cloud of some convex polygon $P$, 
and to return $P$ if this is the case.
It consists of two passes through $\maxOmega$, which are detailed below. 
During the first pass we compute a list $S$ of all strictly narrow vertices of $P$ that are not hidden pivots. With each such vertex $u$, 
we store the supporting lines of the two edges of $P$ incident to $u$. In the second pass we use this list to reconstruct $P$.

\textbf{First pass.} 
We iterate through the pivots of $\maxOmega$ (recall that these are not hidden by the definition of $\maxOmega$). 
For the currently processed pivot $u$, we maintain the point $v$ on $\maxOmega$ such that $\turnValue{v,u} = 2(\pi-\omega)$. 
If pivot $u$ is narrow, we jump to the point on $\maxOmega$ at the distance $2(\pi-\omega)$ from $u$. 
Moreover, if $u$ is strictly narrow, we add  $u$ to the list $S$. 
If $u$ is not narrow, we process the next pivot of $\maxOmega$. We now give the details.

Let $\Gamma$ be the arc of $\maxOmega$ incident to $u$ and following it in clockwise direction. Let $\Gamma_r$ be
the arc following $\Gamma$, and $C_r$ be the supporting circle of $\Gamma_r$. 
We consider several cases depending on the angular measure $|\Gamma|$ of $\Gamma$: 

\begin{itemize}
\item[(a)] 
$|\Gamma| < 2(\pi-\omega)$. See Figure~\ref{fig:first-pass}.
\begin{itemize}
\item[(i)] 
Circle $C_r$ passes through $u$, see Figure~\ref{fig:first-pass}a. Then $u$ is narrow by Lemma~\ref{lemma:narrow-char}b. 
By tracing $\maxOmega$, find the point $w$ on it such that $\turnValue{u,w} = 2(\pi-\omega)$. 
In case $u$ is strictly narrow (i.e., $\angle vuw < \omega$), add $u$ to the list $S$ with the lines through $vu$ and $uw$ (these are the intended lines due to Lemma~\ref{lemma:turns}). 
Set $v$ to be $u$, and $u$ to be $w$ (regardless the later condition). 

\item[(ii)] 
Circle  $C_r$ does not pass through $u$, see Figure~\ref{fig:first-pass}b. Then $u$ is not narrow
by Lemma~\ref{lemma:narrow-char}a. 
Set $u$ to be the other endpoint of $\Gamma$, and update  $v$ accordingly.
\end{itemize} 
\begin{figure}
\begin{minipage}{0.49\linewidth}
\centering
    \includegraphics[page=7,scale = 0.9]{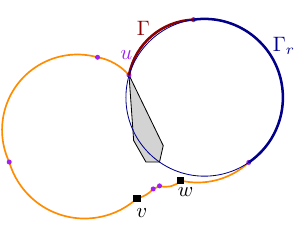}
    \\
    (a)
\end{minipage}
\hfill
\begin{minipage}{0.49\linewidth}
\centering
    \includegraphics[page=6,scale = 0.9]{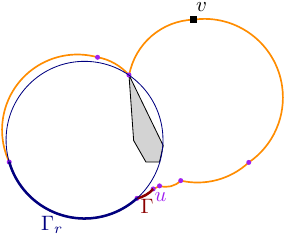}
    \\
    (b)
\end{minipage}
    \caption{First pass of the $\omega$-aware algorithm: pivot $u$ is (a) strictly narrow, (b) non-narrow}
    \label{fig:first-pass}
\end{figure}
 
\item[(b)] $|\Gamma| = 2(\pi-\omega)$. Then $u$ is narrow by Lemma~\ref{lemma:narrow-char}b. 
Let $w$ be the other endpoint of $\Gamma$. Update $S$, $v$, and $u$ as in item a(i). 

\item[(c)] $|\Gamma| = 2t(\pi-\omega)$ for some integer $t>1$. Then $\Gamma$ is in fact multiple arcs separated by hidden pivots, see Lemma~\ref{lemma:hidden-pivots} and Corollary~\ref{cor:max-measure}.  Let $p$ be the other endpoint of $\Gamma$.
Let $w$ and $w'$ be the points on $\Gamma$ such that $\turnValue{u,w} = 2(\pi-\omega)$ and $\turnValue{w',p} = 2(\pi-\omega)$. 
 Update $S$ as in item a(i). Set $u$ to be $p$, and $v$ to be $w'$.
\item[(d)] Otherwise, stop and report that $\maxOmega$ is not the maximal $\omega$-cloud of any polygon. 
\end{itemize}

Note that before starting the above procedure, 
we need to find the starting positions $v$ and $u$. 
We choose $v$ arbitrarily and traverse $\maxOmega$ until we reach the corresponding position of $u$ (i.e., $\turnValue{v,u} = 2(\pi-\omega)$). 
Since we already have traversed the portion of $\maxOmega$ between $v$ and $u$, 
in order to perform exactly one pass through $\maxOmega$, we will have to finish the above procedure as soon as 
the pointer $u$ has reached the initial position of $v$ (not the initial position of $u$). 
However, this is still enough for creating the complete list $S$. Indeed, 
by Lemma~\ref{lemma:turns}b, there is at most one narrow pivot between $v$ and $u$. 
There is no such narrow pivot if and only if the left arm of the  minimal $\omega$-wedge $W_r(v)$ and the right arm of $W_\ell(u)$ coincide, and in this case
there is nothing to add to $S$. 
If there is such a narrow pivot $w$, then neither $u$ nor $v$ is narrow. 
By Lemma~\ref{lemma:turns}c,  
both the left arm of the wedge $W(v)$ and the right arm of $W(u)$ pass through $w$.  Thus $w$ can be found 
 as their intersection.

\textbf{Second pass.} In case list $S$ is empty, we 
apply the procedure of Lemma~\ref{lemma:reconstr-no-narrow} to the whole $\maxOmega$. 
In particular,  as both the start and the endpoint,  
we take the point $x$ with which we completed the first pass of the algorithm; the point  $x'$ such that $\turnValue{x',x} = 2(\pi-\omega)$ is already known from the first pass.
Then $\dirr{x} = \dir{x}$ is the direction of the  minimal $\omega$-wedge with the apex at $x$ and the right arm passing through $x'$.

Suppose now that list $S$ contains $k$ vertices.  They  subdivide $\maxOmega$ into $k$ connected portions that are free from strictly narrow non-hidden pivots. Each portion is treated as follows: 
\begin{itemize}
\item
 If it is a single maximal arc of measure $2t(\pi-\omega)$, we simply separate it by $t-1$ equidistant points, and those points are exactly the vertices of the portion of $P$ corresponding to the considered component of $\maxOmega$, see Lemma~\ref{lemma:hidden-pivots}.

\item Otherwise it is a portion free from any strictly narrow pivots. We process this portion by the procedure of Lemma~\ref{lemma:reconstr-no-narrow}. 
\end{itemize}

Correctness and complexity of 
this algorithm
are analyzed in the following theorem.

\begin{theorem}
\label{thm:algo-given-w}
Given an angle  $\omega$ such that 
$0 < \omega < \pi$, and a circular arc sequence $\maxOmega$ of $n$ arcs with $n>1$, 
there is an algorithm to check if $\maxOmega$ is the maximal $\omega$-cloud of some 
convex polygon $P$, 
and to return $P$ if this is the case.
 The algorithm works in $O(n)$ time, making two passes through the input, and it 
uses $O(1)$ working space. The constants in big-Oh depend only on $\omega$. 
\end{theorem}

\begin{proof}
Consider the first pass of the algorithm. 
Cases a(i), b, and c are exactly the cases where $u$ is a narrow pivot, see Lemmas~\ref{lemma:narrow-char} and \ref{lemma:hidden-pivots}. 
The case d corresponds to an impossible situation
by Lemma~\ref{lemma:hidden-pivots} and Corollary~\ref{cor:max-measure}.
In case a(i), we update $u$ to be $w$, thus we are not processing the pivots between $u$ and $w$; Lemma~\ref{lemma:turns} guarantees 
we do not miss any narrow pivot.
Correctness of the second pass is  implied by Lemmas~\ref{lemma:hidden-pivots} and \ref{lemma:turns}a.

Each pivot of $\maxOmega$ is visited exactly once during the first pass (it is either checked as the pivot $u$,  traversed in the case a(i), or skipped in the case c). The second pass as well visits each pivot once. The time spent during each such visit is $O(1 + t)$, where $t$ is the number of hidden pivots on the maximal arc adjacent to the visited pivot. 
The storage required by the algorithm is the storage required for the list $S$, i.e., it is $O(k)$.  

Observe finally that by Lemma~\ref{lemma:turns}, for any pair $u,v$ of narrow pivots of $\maxOmega$, $\turnValue{u,v} \geq 2(\pi-\omega)$. By Corollary~\ref{cor:total-measure} the total angular measure of the arcs of  $\maxOmega$ is at most $4\pi$. Thus the total number of narrow pivots  of $\maxOmega$  (including the hidden pivots) is at most $\floor{2\pi/(\pi-\omega)}$, which is a constant if $\omega$ is fixed. This completes the proof.     
\end{proof}

\begin{remark}
It is not difficult to modify the above algorithm to work for the easier setting, where the input is the proper $\omega$-cloud of $P$ 
instead of its maximal version.
\end{remark}

\subsection{An  \texorpdfstring{$\omega$}{omega}-oblivious reconstruction algorithm}\label{sec:reconstr-without-w}
In this section we assume that $\omega < \pi/2$. Thus there are no hidden pivots in
the $\omega$-cloud of $P$, and therefore the input sequence equals the (not maximal) $\omega$-cloud
 $\Omega$. 
 We now give the reconstruction algorithm, which is the main result of this section.

\begin{theorem}
\label{thm:vertex-only}
Given a circular arc sequence $\Omega$, there is an algorithm that finds the convex 
polygon $P$ such that $\Omega$ is the $\omega$-cloud of $P$ for some angle $\omega$ with $0 < \omega < \pi/2$, if such a polygon exists. 
Otherwise, it reports that such a polygon does not exist.
 The algorithm works in $O(n)$ time, making two passes through the input, and it 
uses $O(1)$ working space. 
\end{theorem}

\begin{proof}
If $\Omega$ is a single arc, i.e., it is a circle, we return the negative answer, as a circle cannot be an $\omega$-cloud for $\omega < \pi/2$ by Lemma~\ref{cor:circular-cloud}. Below we assume that $\Omega$ has at least two arcs. 

If the total angular measure of $\Omega$ is $4\pi$, then no pivot is strictly narrow due to Corollary~\ref{cor:total-measure}. Starting from any pivot $u$ of $\Omega$ and applying Lemma~\ref{lemma:normal-pivot}, 
we find the placement of one of the arms of the minimal $\omega$-wedge with the apex $u$. Since no pivot is strictly narrow, Lemma~\ref{lemma:turns}c determines whether it is the right or the left arm. This gives us the angle and the placement of the wedge. We can perform the procedure of Lemma~\ref{lemma:reconstr-no-narrow} to reconstruct the polygon.

If the total  angular measure of $\Omega$ is less than $4\pi$, there must be a strictly narrow pivot. We perform a pass through $\Omega$ to find a 
maximal sequence of arcs whose supporting circles contain 
the same pivot $u$, and such that the sequence starts at $u$. 
If we found such sequence of at least two arcs, the distance between $u$ and the endpoint of this sequence is $2(\pi-\omega)$ due to Lemma~\ref{lemma:narrow-char} and Observation~\ref{obs:narrow}. This determines the value of $\omega$. We run the  $\omega$-aware reconstruction procedure of Section~\ref{sec:reconstr-given-w}. If there is no such sequence of at least two arcs, then by Observation~\ref{obs:narrow}  all pivots must be narrow, but  since $\omega< \pi/2$, the polygon $P$ in that case must be a line segment. Whether $\Omega$ is an $\omega$-cloud of a line segment for some $\omega$ can be checked in one pass through $\Omega$.

We can combine finding the above maximal sequence with counting the total angular measure of $\Omega$ in one pass; this pass requires $O(1)$ working  storage. The claim then follows from the Theorem~\ref{thm:algo-given-w}.  
\end{proof}

\bibliography{OmegaCloud}

\begin{thebibliography}{10}
\providecommand{\url}[1]{\texttt{#1}}
\providecommand{\urlprefix}{URL }

\bibitem{abellanas2011coverage}
Abellanas, M., Bajuelos, A., Hurtado, F., Matos, I.: Coverage restricted to an
  angle. Operations Research Letters  39(4),  241--245 (2011)

\bibitem{ALOUPIS2010115}
Aloupis, G., Cardinal, J., Collette, S., Hurtado, F., Langerman, S., O'Rourke,
  J., Palop, B.: Highway hull revisited. Comput. Geom.  43(2),  115--130 (2010)

\bibitem{DBLP:journals/comgeo/BoseCSS16}
Bose, P., Carufel, J.D., Shaikhet, A., Smid, M.: Probing convex polygons with a
  wedge. Comput. Geom.  58,  34--59 (2016)

\bibitem{bmss11}
Bose, P., Mora, M., Seara, C., Sethia, S.: On computing enclosing isosceles
  triangles and related problems. Int. J. Comput. Geom. Ap.  21(01),  25--45
  (2011)

\bibitem{DBLP:journals/amai/BrucksteinL91}
Bruckstein, A., Lindenbaum, M.: Reconstruction of polygonal sets by constrained
  and unconstrained double probing. Ann. Math. Artif. Intell.  4,  345--361
  (1991)

\bibitem{DBLP:journals/jal/ColeY87}
Cole, R., Yap, C.: Shape from probing. J. Algorithms  8(1),  19--38 (1987)

\bibitem{Dobkin:1986:PCP:12130.12174}
Dobkin, D., Edelsbrunner, H., Yap, C.K.: Probing convex polytopes. In: Proc.
  STOC'86. pp. 424--432. ACM (1986)

\bibitem{DBLP:journals/siamcomp/EdelsbrunnerS88}
Edelsbrunner, H., Skiena, S.: Probing convex polygons with x-rays. {SIAM} J.
  Comput.  17(5),  870--882 (1988)

\bibitem{DBLP:conf/isaac/FleischerW09}
Fleischer, R., Wang, Y.: On the camera placement problem. In: Proc. ISAAC'09.
  pp. 255--264 (2009)

\bibitem{Gardner92}
Gardner, R.J.: X-rays of polygons. Discrete Comput. Geom.  7,  281--293 (1992)

\bibitem{DBLP:journals/ipl/Li88a}
Li, S.: Reconstruction of polygons from projections. Inf. Process. Lett.
  28(5),  235--240 (1988)

\bibitem{MeijerSkiena96}
Meijer, H., Skiena, S.: Reconstructing polygons from x-rays. Geometriae
  Dedicata  61,  191--204 (1996)

\bibitem{moslehi2017separating}
Moslehi, Z., Bagheri, A.: Separating bichromatic point sets by minimal
  triangles with a fixed angle. Int. J. Found. Comput. S.  28(04),  309--320
  (2017)

\bibitem{DBLP:journals/ijrr/RaoG94}
Rao, A., Goldberg, K.: Shape from diameter: Recognizing polygonal parts with a
  parallel-jaw gripper. I. J. Robotic Res.  13(1),  16--37 (1994)

\bibitem{DBLP:journals/algorithmica/Skiena89}
Skiena, S.: Problems in geometric probing. Algorithmica  4(4),  599--605 (1989)

\bibitem{DBLP:journals/jal/Skiena91}
Skiena, S.: Probing convex polygons with half-planes. J. Algorithms  12(3),
  359--374 (1991)

\end{thebibliography}
\end{document}